\documentclass[conference]{IEEEtran}
\IEEEoverridecommandlockouts
\usepackage{psfrag}
\usepackage{graphicx}
\usepackage{dsfont}
\usepackage{color}
\usepackage{url}
\usepackage{amsfonts}
\usepackage{amsthm}
\usepackage{amsmath}
\usepackage{amssymb}

\usepackage{tikz}
\usetikzlibrary{arrows}
\usetikzlibrary{decorations.markings,calc}

\usepackage{cite}
\graphicspath{{figs/}}
\usepackage{pseudocode}

\newtheorem{theorem}{Theorem}
\newtheorem{lemma}{Lemma}
\newtheorem{remark}{Remark}
\newtheorem{definition}{Definition}
\newtheorem{corollary}{Corollary}

\newtheorem{proposition}{Proposition}

\usepackage{enumerate}
\newcommand{\EE}{\mathbb{E}}
\newcommand{\enc}{\mathsf{c}}
\newcommand{\supp}{\mathsf{supp}}
\begin{document}

\author{
\IEEEauthorblockN{Ashwin~Pananjady~and~Thomas~A.~Courtade}
\IEEEauthorblockA{Department of Electrical Engineering and Computer Sciences\\
			     University of California, Berkeley\\
              Email: \{ashwinpm, courtade\}@eecs.berkeley.edu \vspace{-5ex}}
\thanks{This work was supported in part by the NSF Center for Science of Information CCF-0939370.}
}

\title{\vspace{-9pt} Compressing Sparse Sequences \\under Local Decodability Constraints \vspace{-3pt}}
\maketitle
\begin{abstract}
We consider a variable-length source coding problem subject to local decodability constraints.  In particular, we investigate the blocklength scaling behavior attainable by encodings of $r$-sparse binary sequences, under the constraint that any source bit can be correctly decoded upon probing at most $d$ codeword bits.  We consider both adaptive and non-adaptive access models, and derive upper and lower bounds that often coincide up to constant factors.  Notably, such a characterization for the fixed-blocklength analog of our problem remains unknown, despite considerable research over the last three decades.  Connections to communication complexity are also briefly discussed.

\end{abstract}

\vspace{-3pt}

\section{Introduction} \label{sec:intro}
Efficient representation of  a sequence of \emph{source bits} by a significantly shorter sequence of \emph{encoded bits} (i.e.,  a \emph{codeword})  is the classical problem of lossless source coding, proposed by Shannon  in his seminal 1948 paper \cite{shannon}. It is widely known that  optimal compression performance can be achieved with  schemes such as Huffman codes \cite{huffman} or the Lempel-Ziv universal compression algorithm \cite{lzcomp}. However, these compression schemes suffer from the drawback that they do not support \emph{local decodability}. Specifically,  retrieving  a single bit of the source sequence generally requires a decoder to access \emph{all} of the encoded bits. 

This is clearly undesirable in applications that favor retrieving selected pieces of information, rather than the entire source sequence. One such application is in bioinformatics \cite{pavlichin2013human, deorowicz2013genome}, in which a DNA sequence is stored as a binary string with relation to a \emph{reference} sequence, with $1$s representing single nucleotide polymorphisms (SNPs) at those positions.
In SNP calling, we are interested in learning whether there is a SNP at position $i$.  Since we are not interested in any other information about the sequence, we would ideally like to accomplish this by accessing  few bits in the compressed representation of the DNA sequence. In this specific instance, local decodability is strongly motivated, since decompressing the whole genome can be prohibitively expensive from a memory standpoint.

Another example presents itself in  the efficient storage of relationships among objects (e.g., relational databases \cite{relational}).  Given a collection of $n$ objects, the relationships among these objects can be represented by an undirected graph on $n$ vertices, with the presence or absence of an edge $(i,j)$ signifying that objects $i$ and $j$ are related or unrelated, respectively (e.g., friendships in social networks). One can think of representing all graphs with $n$ vertices by sequences of $\binom{n}{2}$ bits representing all possible edges. A `$1$' at a given position indicates the presence of that edge, and a `$0$' indicates that it is absent. Thus, testing relationships between objects is accomplished by querying the value of the corresponding bit.  As in SNP calling, it would be ideal to have a compressed representation of the graph which permits such queries upon accessing a small number of encoded bits.

We remark that both applications referred to above involve a source that is inherently sparse -  both SNPs and the number of relationships are small compared to the total length of the sequence.  
Motivated in part by this, our objective in this paper is to analyze the fundamental tradeoffs between access constraints and compressibility of sparse sequences, in the context of locally decodable compression. We  consider a variable blocklength model, in which  source sequences can be mapped to codewords of varying lengths, and that the decoder is informed of the codeword length at the start of the decoding process.

Prior work on the problem of locally decodable source coding includes results on succinct data structures in the \emph{bit probe} \cite{elias, 1993bit} and \emph{cell probe} \cite{cellprobe, yao} complexity models. %
A widely studied problem in the bit-probe model is the static membership problem, which is closely related to the problem we consider. The bit-probe model for the static membership problem encodes subsets $S\subseteq \{1,2, \ldots, n\}$ of size at most $r$ into a data structure of \emph{fixed} length $\ell$, such that queries of the form ``Is $i\in S$" for  $i\in \{1,2, \ldots, n\}$ can be determined by probing (i.e., accessing) at most $d$ bits in the data structure, either adaptively or non-adaptively. Buhrman et al.\ \cite{bitvector} provided the lower bound  $\ell = \Omega(dr^{1-1/d}n^{1/d})$, which remains the best lower bound for a general $n,r,d$. They also showed a scheme that achieves a blocklength $\ell =O(rd'n^{1/d'})$, with $d' =d-\Theta(\log r+\log \log n)$ and $d > \Theta(\log r + \log \log n)$. The interested reader is referred to \cite{nicholson2013survey} for a comprehensive survey on several improvements to these bounds (for specific regimes of $r$ and $d$) that have been proposed in the literature \cite{alon2009power, radhakrishnan2001explicit, viola2012bit, lewenstein2014improved}. Notably, \cite{alon2009power} considered $d\leq 4$ probes and showed that, for $r=o(n)$, a blocklength $\ell = o(n)$ can be achieved by schemes using $3$ adaptive probes or $4$ non-adaptive probes, thereby settling a question posed in \cite{bitvector}. By letting $S$ denote the set of indices where a binary sequence has ones, the static membership problem considered by the bit probe model yields a fixed-blocklength, locally decodable representation of sparse sequences, and can therefore be viewed as a fixed-blocklength counterpart to the problem that we consider. %

Largely independent from the prior work on the bit probe model, locally decodable source coding has also received recent attention from the information theory community \cite{chandar, makhdoumi2013locally, zhou2014simple}. Closely related to the present paper is the recent work by Makhdoumi et al. \cite{makhdoumi2013locally}, which considers the design of locally decodable source codes under the bit probe model for i.i.d.\ Bernoulli sources, with vanishing block-error probability.

Our model differs fundamentally from those appearing in both \cite{bitvector} and \cite{makhdoumi2013locally} since we consider a variable-blocklength setting (to be defined precisely in Section \ref{sec:ProblemSetting}), which has not been previously studied. This  is practically motivated because compressed file size is rarely fixed a-priori by the compression scheme, and file length is often recorded in metadata available to a decompressor.  As we show in the sequel, our analysis of the variable blocklength case allows us to provide tight order-wise bounds on the average blocklength of the code in many cases --- a problem that has remained open for the fixed-blocklength bit-probe problem for over 3 decades. Also, in contrast to \cite{makhdoumi2013locally}, we restrict ourselves exclusively to the \emph{lossless} setting, in which the decoder must have zero (and not vanishing) probability of error, which is motivated by high-fidelity applications such as SNP calling.

\vspace{-3pt}

\subsection*{Our Contributions }

In this paper, we give  upper and lower bounds on the average blocklength attainable by  variable-blocklength compression schemes under local decodability constraints. These bounds are non-asymptotic in nature, and  coincide (up to constant factors) in many cases. As a corollary, we give necessary and sufficient conditions on the number of bit-probes required to achieve competitively optimal compression performance, and briefly comment on connections to communication complexity. %

\section{Notation and Problem Setting}\label{sec:ProblemSetting}
For an integer $k\geq 1$, we employ the shorthand notation    $[k]\triangleq \{1,2,\dots, k\}$.  We make frequent use of the conventional notations $O(\cdot), o(\cdot), \Omega(\cdot), \omega(\cdot), \Theta(\cdot)$.

 Throughout, we  consider encodings of  $r$-sparse binary vectors, which are simply sequences $x^n =(x_1,x_2, \dots, x_n)\in \{0,1\}^n$ having Hamming weight precisely $r$ (we may assume without loss of generality that $r\leq n/2$).  Our restriction to sequences of weight precisely $r$ is primarily for convenience, since our arguments readily generalize to vectors having weight at most $r$.  In some cases, we allow the sparsity parameter $r$ to scale with $n$, in which case we write $r_n$.
 
  The \emph{support} of a source sequence $x^n$ is defined to be the set of nonzero coordinates, i.e. $\mathsf{supp}(x^n) = \{i\in [n]: x_i = 1\}$.  When referring to multiple distinct  sequences, we will use the bracket subscript notation, i.e. $x^n_{(1)}, x^n_{(2)}, \cdots$.

  Letting ${[n]\choose r} \subset \{0,1\}^n$ denote the set of $r$-sparse binary vectors, we assume random vectors $X^n\in {[n]\choose r} $ are drawn uniformly from all ${n \choose r}$ possibilities. 
A source code (i.e., compressor) $\enc$ for $r$-sparse vectors is an invertible mapping $\mathsf{c} \colon {[n]\choose r}  \to \{0,1\}^*$, where $\{0,1\}^*=\{0,1,00,01,10,\dots\}$ denotes the set of all binary strings.  For a codeword $\enc(x^n) = (c_1, c_2, \dots, c_{\ell})$ and a set of integers $S\subset \{1,2,\dots\}$, we let $\enc_S(x^n) = \{c_i : i\in S\}$ denote the coordinates of the codeword indexed by entries in $S$. Letting $\ell(b)$ denote the length of $b\in \{0,1\}^*$, we remark that there are source codes for which the average codeword length is roughly\footnote{Here and throughout, $\log(x)$ denotes the base-2 logarithm of $x$.}  \vspace{-2pt}
\begin{align}
\EE \left[ \ell(\mathsf{c}(X^n))\right] \approx \log \binom{n}{r} \mathsf{~bits}, \label{ITlowerbound}
\end{align}
and this is best-possible, since the entropy of the source $H(X^n) = \log \binom{n}{r}$.  Indeed, the na\"ive scheme which lists the positions of each nonzero entry (requiring approximately $\log n$ bits each) is essentially optimal when $r\ll n$.  However, it is not clear whether such a source code admits a decoding algorithm that, for any specified index $j\in[n]$, can recover bit $x_j$ by probing a \emph{bounded} number of bits in $\enc(x^n)$.    Thus, in the spirit of locally-decodable error-correcting codes \cite{katz2000efficiency} and the data structure counterparts in \cite{bitvector, makhdoumi2013locally}, we define a \emph{variable-length $(r,d,n)$-locally decodable source code}:
\begin{definition}\label{def:LocDec}
A  {$(r,d,n)$-locally decodable} source code, or simply, an $(r,d,n)$ code, consists of a mapping   \vspace{-2pt}
\begin{align}
\enc : {{[n]}\choose r} \to \{0,1\}^*
\end{align}
with the property that, for each $x^n$ and $j\in [n]$, there exists a set $S\subseteq \{1,2,\dots\}$ of size $|S|\leq d$ for which $x_j$ is a function of $\ell(\enc(x^n))$ and $\enc_S(x^n)$.
\end{definition}

In  other words, we can say $\enc$ is a $(r,d,n)$-locally decodable source code only if there exists a corresponding `$(r,d,n)$-local decompressor' --- i.e.,  an algorithm that takes as input a \emph{query index} $j\in [n]$ and the codeword length $\ell({\enc(x^n)})$, and returns the data bit $x_j$ after  accessing at most $d$ bits of $\enc(x^n)$. In light of this, we refer to the number $d$ as an \emph{access constraint} (or, decoding depth), since it bounds the number of encoded bits that the decoder probes before making a determination.  In contrast to the fixed-blocklength settings that have been considered previously (cf. \cite{katz2000efficiency, bitvector, makhdoumi2013locally}), Definition \ref{def:LocDec} does not preclude variable-length encoding schemes. As mentioned above, this is motivated by practice, where data structures are usually of variable length and any access protocol is cognizant of the encoded data's length so that segmentation faults are avoided.  Indeed, in computer file systems, a file is typically accessed after first reading  metadata that describes the location and length of the file.

Note that our definition of an $(r,d,n)$-local decompressor does not distinguish between adaptive or non-adaptive bit probes.  That is, a decompressor can probe entries of ${\enc(x^n)}$ in an adaptive manner (where codeword locations are accessed sequentially, and the positions accessed can depend on the bit values observed during previous probes),  or in a non-adaptive manner (where codeword locations accessed are determined only by the {query index} $j\in [n]$ and the codeword length $\ell({\enc(x^n)})$).  When such a distinction is necessary, we will explicitly refer to adaptive and non-adaptive $(r,d,n)$ codes.

\section{Main Results}\label{sec:MainResults}

\subsection{Bounds on  expected blocklength} \label{sec:explength}
In this section, we present lower and upper bounds on the expected blocklength achievable by variable-length source codes obeying a local decodability constraint, and give sufficient conditions for them to coincide.  Proofs can be found in Section  \ref{sec:proofsketch}.%

\begin{theorem}\label{lb:ad}
The expected codeword length of any $(r,d,n)$-locally decodable code with adaptive  bit-probes satisfies
\begin{align}
   \EE \left[ \ell(\mathsf{c}(X^n))\right] +1&\geq \left(\frac{rd+1}{4e}\right) \left({n \choose r}^{1/(rd+1)}-1\right). \label{nonAsympConverse}
\end{align}
\end{theorem}

Recalling the identity $\lim_{m\to \infty} m (x^{1/m}-1) = \ln x$, we note that for fixed $n,r$  the lower bound \eqref{nonAsympConverse} becomes
\begin{align}
\lim_{d\to \infty} \left({n \choose r}^{1/(rd+1)}-1\right)\left(\frac{rd+1}{4e}\right) = \frac{1}{4e} \ln {n \choose r}. \label{newConverse}
\end{align}
Hence, we recover the information-theoretic lower bound \eqref{ITlowerbound} (up to constant factors) in the absence of a local decodability constraint.  %
 On this note, an  important consequence of Theorem \ref{lb:ad} is that it  dictates how quickly $d$ must scale with respect to $n,r$ in order to accommodate encoding schemes that are near-optimal in the information-theoretic sense.  
In the next section we quantify this tension more precisely, and establish how large $d$ must be in order to ensure competitive optimality.  Before doing this, we discuss the tightness of \eqref{nonAsympConverse}.

\begin{theorem}\label{ub:nonad}
For any choice of $r,d,n$, there exists a non-adaptive $(r,d,n)$-locally decodable code $\enc$ with average codeword length
\begin{align}
\EE\left[ \ell(\mathsf{c}(X^n))\right] \leq 30 (rd+1)  \left((r+1)^{(r+1)} {n \choose r }\right)^{1/(rd+1)}  \!\!\!\!\!\!  .\label{nonAsympAch}
\end{align}
\end{theorem}

Two remarks are in order.  First, we  emphasize that Theorem \ref{lb:ad} is a converse result for adaptive schemes, while Theorem \ref{ub:nonad} is an achievability result for non-adaptive schemes.  We will see shortly that these bounds coincide (up to constant factors) in many cases, showing that adaptivity provides at most constant-factor improvement in these settings.  Second, we note that both Theorem \ref{lb:ad} and Theorem \ref{ub:nonad} are non-asymptotic in nature.  That is, they hold for any choice of parameters $r,d,n$.  However,  results become most crisp when $n\to \infty$, and $r,d$ are functions of $n$.  As a first example, we take $n\to \infty$ and $r,d$ fixed (i.e., not depending on $n$).  In this case, we find that the blocklength of an optimal sequence $\{\enc^{\star}_n\}$ of $(r,d,n)$ codes scales as  
$\EE\left[ \ell(\mathsf{c}_n^{\star}(X^n))\right] = \Theta ( n^{r/(rd+1)})$. 
Hence, when $r,d$ are fixed,  performance scales poorly relative to the information-theoretic lower bound of $\Theta(\log n)$.

As a second example, consider the setting where $r_n=n^{\epsilon}$ and $d_n=\delta \log n$.  Then it is a straightforward calculation using \eqref{nonAsympConverse} and \eqref{nonAsympAch} to see that 
any optimal sequence $\{\enc^{\star}_n\}$ of $(r_n,d_n,n)$-locally decodable codes will satisfy
\begin{align}
C_1  (2^{(1-\epsilon)/\delta} -1)   \leq \frac{\EE\left[ \ell(\mathsf{c}_n^{\star}(X^n))\right]}{\delta n^{\epsilon} \log n }  \leq C_2  2^{1/\delta}   \mbox{~~as $n\to \infty$},\notag
\end{align}
where $C_1$ and $C_2$ are absolute constants.  Thus, up to constant factors, the blocklength scaling behavior  of optimal codes in this regime is $\delta n^{\epsilon}\log n$, and the decoder will probe a fraction of the codeword proportional to $1/r_n$ in worst case. Contrast this with the trivial encoding scheme that simply stores the position of each `1'; the natural decoder based on binary search would require roughly  $\log(r_n) \cdot \log(n)$ probes in worst case.

Similarly, if we parameterize $n_m = {m \choose 2}$, $r_m = (1+\epsilon)\frac{\ln m}{m} {m \choose 2}$ and $d_m = \delta \log m$, then as $m\to \infty$ any optimal sequence $\{\enc^{\star}_m\}$ of $(r_m,d_m,n_m)$-locally decodable codes will satisfy \vspace{-7pt}
\begin{align}
C_1 (2^{1/\delta}-1)   \leq \frac{\EE\left[ \ell(\mathsf{c}_m^{\star}(X^{n_m}))\right]}{r_m d_m}  \leq C_2 2^{2/\delta}.
\end{align}
This particular choice of parameters can be interpreted as encoding a random graph on $m$ vertices with $(1+\epsilon)\frac{\ln m}{m} {m \choose 2}$ edges.  Since $\frac{\ln m}{m}$ is the threshold for connectivity, this graph is connected with high probability for $\epsilon>0$.  Now, querying whether two vertices are connected in this graph corresponds to querying a bit of $X^{n_m}$.  In order to accomplish this in time that grows logarithmically in the number of vertices requires average blocklength of order $r d = \Theta\left( m{\log^2(m)} \right)$.

In the latter two examples,  average blocklength scales  $r_n d_n = \Theta(\log {n \choose r_n })$, which is within constant factors  of the information-theoretic lower bound (i.e., competitively optimal).  In both cases, we chose $r_n d_n= \Omega( \log {n\choose r_n})$ in order to achieve this scaling.  Thus, it is natural to ask:  \emph{do there exist competitively optimal schemes with $r_n d_n = o(\log {n\choose r_n})$?}  The answer to this question is negative, and is the focus of the next section. 
However, before we proceed, we unify the above examples under  the following straightforward corollary of Theorems \ref{lb:ad} and \ref{ub:nonad}:
\begin{corollary}\label{cor:MatchBounds}
If $\log {n \choose r_n} = \Omega(r_n d_n)$ and $d_n=\Omega (\log r_n)$, then any optimal sequence of $(r_n,d_n,n)$-locally decodable codes $\{\enc^{\star}_n\}$ satisfies \vspace{-7pt}
\begin{align}
\EE \left[ \ell(\mathsf{c}^{\star}_n(X^n))\right] = \Theta\left(  r_n d_n  {n \choose r_n}^{1/(r_n d_n+1)} \right) . 
\end{align}
\end{corollary}
\subsection{Local Decodability and Competitive Optimality} \label{sec:ldco}

We now focus on the question raised at the end of the previous section, and give necessary conditions for competitive optimality (proofs can be found in the supplementary material). To this end, we define:

\begin{definition}%
For a sequence of integers $\{r_n\}_{n\geq 1}$ a  sequence of encoders \vspace{-10pt}
\begin{align}
\enc_n : {{[n]}\choose r_n}\to \{0,1\}^*~~~~n\geq 1
\end{align}
 is said to be \emph{competitively optimal} if %
\begin{align}
\limsup_{n\to \infty} \frac{\EE[\ell(\enc_n(X^n))]}{\log \binom{n}{r_n}} = O(1).
\end{align}
\end{definition}
In other words, competitively optimal schemes attain compression rates within a constant factor of the information theoretic lower bound $\log {n \choose r_n}$ for large enough $n$.%

From Theorem \ref{lb:ad}, it is possible to deduce the following necessary condition for competitive optimality:
\begin{theorem}\label{thm:compOptimalNecessary}
If $\{\enc_n\}$ is a competitively optimal sequence of $(r_n,d_n,n)$-locally decodable codes, then  
$r_n d_n = \Omega \left( \log {n \choose r_n}\right)$.%
\end{theorem}
 In other words, we cannot expect to attain competitive optimality when $r_n$ and $d_n$ are simultaneously small relative to the source entropy 
(note the contrast to the sufficient conditions in Corollary \ref{cor:MatchBounds}).  
This relationship can be somewhat complicated since the source entropy generally depends on both $n$ and $r_n$.  However, when the source sequence  is modestly sparse (i.e., $r_n = O(n^{1-\epsilon})$ for some $\epsilon>0$), then the explicit dependence on $r_n$ in Theorem \ref{thm:compOptimalNecessary} %
can be eliminated to obtain the following condition:
\begin{corollary}\label{cor:CompOptimalSparse}
If $r_n = O(n^{1-\epsilon})$ for some $\epsilon>0$, then there exists a competitively optimal sequence of $(r_n,d_n,n)$-locally decodable codes if and only if $d_n = \Omega (\log n)$.
\end{corollary}

In contrast to Corollary \ref{cor:CompOptimalSparse}, if $r_n=\Theta(n)$, the information theoretic lower bound is $\log {n \choose r_n} = \Theta(n)$, and the identity encoding $\enc_n(x^n) = x^n$ is competitively optimal, with all source bits being decodable with $d_n=1$ probes.

\section{Proof Sketches for Main Results}\label{sec:proofsketch}
Due to space constraints, we only sketch the proofs of Theorems \ref{lb:ad} and \ref{ub:nonad}.  Details are deferred to the supplementary material along with the proofs for Section \ref{sec:ldco}. 

 \vspace{2pt}

\begin{IEEEproof}[Proof Sketch of  Theorem \ref{lb:ad}]
Our proof is inspired by that of \cite[Theorem 6]{bitvector}. Let $\enc$ be a $(r,d,n)$-locally decodable code. 
For a source sequence $x^n$, let $\enc_q (x^n)$ denote the $q$th coordinate of the codeword $\enc (x^n)$, and define the set
\begin{align}
T^k_i \triangleq &\big\{\big(q, \enc_q (x^n_{(i)})\big): \ell(\enc(x^n_{(i)})) = k, \text{and location $q$ of $\enc(x^n_{(i)})$} \nonumber \\ 
&\text{~~is accessed to determine $x_j$ for some $j \in \supp(x^n_{(i)})$\}},\nonumber
\end{align}
where we have abused notation slightly by letting $x_j$ denote the $j$th coordinate of sequence $x^n_{(i)}$.  
Note that each $T_i^k$ is a subset of $[k] \times \{0,1\}$ of size at most $rd$, since $|\supp(x^n_{(i)})|=r$ and the decoder makes at most $d$ probes in response to a query. Also note that for $i \neq i'$, $T_i^k \not \subseteq T_{i'}^k$. To see this, assume the contrary, that $T_i^k \subseteq T_{i'}^k$ for some $i\neq i'$. Let the encoded source word be $x^n_{(i')}$. If we now query the value of $x_j$ for $j \in \supp(x^n_{(i)}) \setminus \supp(x^n_{(i')})$, we see that the decoder will make an error, establishing  the contradiction.

Since for fixed $k$, the $T^k_i$s are not subsets of one another, an application of the LYM inequality \cite{LYM} yields 
\begin{align}
\#\{i : \ell(\enc(x^n_{(i)})) = k\} \leq \max_{v \leq rd} \binom{2k}{v}~~\mbox{for each $k$.} \label{maxCodewords}
\end{align}
In light of \eqref{maxCodewords}, the average codeword length must satisfy
\begin{align}
\EE \left[ \ell(\mathsf{c}(X^n))\right] &\geq \sum_{k=1}^{M(n,r,d)} k \frac{ \max_{v \leq rd} \binom{2k}{v}  }{{n \choose r}}, \label{exLBgreedy}
\end{align}
where $M(n,r,d)$ is the largest integer satisfying
\begin{align}
 \sum_{k=1}^{M(n,r,d)+1} \max_{v \leq rd} \binom{2k}{v} > {n \choose r} &\geq \sum_{k=1}^{M(n,r,d)} \max_{v \leq rd} \binom{2k}{v} . \label{firstInequality}
\end{align}
Now define the  probability distribution  \vspace{-2pt}
\begin{align}
Q(k) =  \frac{ \max_{v\leq rd}{{2k} \choose v}  }{{n \choose r}} ~~~\mbox{for $1\leq k \leq M(n,r,d)$}  \vspace{-2pt}
\end{align}
and $Q(M(n,r,d)+1 ) = 1-\sum_{k=1}^{M(n,r,d)}Q(k)$.  Since $Q(k) \leq Q(k+1)$ for $k < M(n,r,d)$ by definition, we can conclude  \vspace{-3pt}
\begin{align}
\EE \left[ \ell(\mathsf{c}(X^n))\right] &\geq \sum_{k=1}^{M(n,r,d)+1} k \cdot Q(k) \geq \frac{M(n,r,d)+1}{2}.\label{lowerBoundByMax}  \vspace{-2pt}
\end{align}
Toward evaluating \eqref{lowerBoundByMax},  we need the following  technical estimate, which is proved in the supplementary material. \vspace{-2pt}
\begin{lemma}\label{lem:techEstimate}
For all $M,v\geq 1$,  \vspace{-2pt}
\begin{align}
\sum_{k=1}^M \max_{i \leq v} {{2k} \choose i } \leq 
2^{v} \frac{\left(M+2 + \frac{v+1}{2e}\right)^{v+1}}{(v+1)!} .\label{eqn:KeyConverseLemma}%
\end{align}
\end{lemma}

Identifying $M\leftarrow M(n,r,d)+1$ and $v\leftarrow rd$ in \eqref{eqn:KeyConverseLemma}, the first inequality in \eqref{firstInequality} can be rearranged to conclude  \vspace{-2pt}
\begin{align}
   M(n,r,d)+3  &\geq \left({n \choose r}^{1/(rd+1)}- 1 \right)\left(\frac{rd+1}{2e}\right).%
\end{align}
Recalling \eqref{lowerBoundByMax} proves the desired inequality.
\end{IEEEproof}

 \vspace{2pt}

\begin{IEEEproof}[Proof Sketch of Theorem \ref{ub:nonad}]
The proof is by a random coding argument, but  it is important to note that standard typicality arguments are not applicable here since they  do not support local decodability.  Briefly, the  idea behind our encoding scheme is to first encode some information about $\mathsf{supp}(\enc(x^n))$ into the codeword length, and then carefully encode the remaining information so that bit $x_j$ can be recovered by computing the binary AND of $d$ encoded bits.  A precise description of the codebook generation and decoding procedure is given below. An example and the analysis are postponed to Appendix \ref{app:thm2}.

{\bf Codebook Construction:}  For $k=rd+1, rd+2, \dots$ choose a subset $S_k\subseteq [n] $ of size $\frac{r}{r+1}\frac{\binom{k}{d}}{\binom{rd}{d}}$ uniformly at random from all such subsets\footnote{Floor and ceiling operators are omitted  for clarity of presentation.}.  For each $j\in S_k$, choose a subset $T_{j,k} \subseteq [k]$ of size $d$ independently and uniformly from all such subsets.  All subsets are made available to both encoder and decoder.

For a sequence $x^n\in {[n] \choose r}$, let $k(x^n)$ denote the smallest integer $k$ such that  the following two conditions hold:
\begin{enumerate}[(C1)]
\item  $\mathsf{supp}(x^n)\subseteq S_k$; and 
\item  $T_{j,k}\not\subseteq \cup_{ i \in \mathsf{supp}(x^n) } T_{i,k}$ for all $j \in S_k \setminus \mathsf{supp}(x^n)$.
\end{enumerate}

{\bf Encoding procedure:}  
A sequence $x^n \in {[n] \choose r}$ is encoded to a codeword $\enc(x^n)$ of length $k(x^n)$ satisfying \vspace{-2pt}
\begin{align}
\mathsf{supp}(\enc(x^n)) =  \cup_{ i \in \mathsf{supp}(x^n) } T_{i,k(x^n)}.
\end{align}
In other words, $x^n$ is encoded to a vector of length $k(x^n)$, which has 1's in all positions $j\in T_{i,k(x^n)}$ if and only if $x_i =1$. 

{\bf Decoding procedure:}  On observing the length of codeword $\enc(x^n)=(c_1, c_2, \dots, c_{\ell(\enc(x^n))})$, determine bit $x_j$ as follows:
\begin{enumerate}[(1)]
\item  If $j \notin  S_{\ell(\enc(x^n))}$, declare $x_j=0$; else 
\item  If $j \in  S_{\ell(\enc(x^n))}$, declare $x_j= \wedge_{i \in T_{j,k(x^n)}} c_i$ , where `$\wedge$' denotes binary AND.
\end{enumerate}

By the nature of the codebook construction, it is clear that the decoder (i) will never make an error; and (ii)  satisfies the non-adaptive $d$-local decodability constraint. The analysis of the average codeword length is omitted due to space constraints, and can be found in Appendix \ref{app:thm2}.  
\end{IEEEproof}

\section{Concluding remarks}\label{sec:conclusion}
We provided bounds for the blocklength scaling behaviour of $(r,d,n)$ locally-decodable codes that are order-wise tight for many regimes of $r$, $d$, and $n$, although determining the tight constant in these bounds is still an open problem. We also showed that in contrast to the fixed blocklength setting (cf. \cite{alon2009power}), adaptivity of probes provides no essential advantage in our regime of variable length source coding. 
In conclusion, we mention two variations on our main results:
\subsection{Compression with block errors}
In  \cite{makhdoumi2013locally}, the authors allow for vanishing block-error probability in decoding.  Although we only considered error-free encodings, the proof of Theorem \ref{lb:ad} readily extends to incorporate block-error probability as follows:  Letting $\hat{x}^n$ denote the decoder's estimate of the sequence $x^n$ given codeword $\enc(x^n)$, the block error rate is defined to be $\Pr\{X^n \neq \hat{X}^n\}$.  Now, Theorem \ref{lb:ad} continues to hold for any $(r,d,n)$ code with block-error rate $\varepsilon$ by simply replacing the quantity $\binom{n}{r}$ with $(1-\varepsilon)\binom{n}{r}$.  Indeed, this follows by considering only those sequences that are correctly decoded and making the same substitution in \eqref{firstInequality} in the proof of Theorem \ref{lb:ad}. 
\vspace{-1mm}
\subsection{Connection to communication complexity}  It is known that the bit-probe model has applications to asymmetric communication complexity \cite{miltersen1995data}.  To draw an analogous   connection to our setting, consider an asymmetric communication complexity model \cite{miltersen1995data} in which Alice (the user) has $i \in [n]$, Bob (the server) has $S\subset [n]$ of size $r$, and they wish to compute the membership function
\begin{align}
f(i,S) =
\begin{cases}
1 \text{ if } i \in S \\
0 \text{ otherwise}.
\end{cases}\label{fdef}
\end{align}

We now enforce that the function $f$ must be computed under a \textsc{SpeedLimit} paradigm, which proceeds as follows. Communication starts with Bob sending a \emph{speed limit} message to Alice consisting of some $z$ bits, which limits the length of any of her messages to $z$ bits. Bob's subsequent messages consist of $1$ bit. After the initial round, Alice and Bob communicate over $d$ rounds\footnote{To be consistent with the rest of the paper, a communication round consists of one message by Alice and a response by Bob.} to evaluate $f$.
The setting arises in practice where a server imposes upload bandwidth limits on users it serves (e.g., to maintain quality or fairness of service).

Note that our scheme in Theorem \ref{ub:nonad} provides a communication protocol to compute $f$ under the \textsc{SpeedLimit} paradigm. Bob is essentially given a source sequence $x^n$, which he stores as $c(x^n)$. Alice is given the index $i$ of the source bit that must be decoded, and must do so by making queries to Bob. Bob begins by sending $\ell(\enc(x^n))$ to Alice, using $\log \ell(\enc(x^n))$ bits. Alice then sends messages $m_j$, $j\in[d]$ of $\log \ell(c(x^n))$ bits each. In response to message $m_j$, Bob sends back $c_{m_j}(x^n)$. Alice then announces $f$ to be the AND of the $d$ bits she has received from Bob.
Therefore, from Theorem \ref{ub:nonad}:
\begin{corollary}
There exists a deterministic communication protocol for computing the function $f$ as in \eqref{fdef} under the \textsc{SpeedLimit} paradigm for which the speed limit $z$ and number of communication rounds $d$ satisfy
\begin{align}
\EE\left[2^z\right] \leq 30 (rd+1)  \left((r+1)^{(r+1)} {n \choose r }\right)^{1/(rd+1)}.
\end{align}
\end{corollary}

\bibliographystyle{IEEEtran}
\bibliography{research}

\begin{thebibliography}{10}
\providecommand{\url}[1]{#1}
\csname url@samestyle\endcsname
\providecommand{\newblock}{\relax}
\providecommand{\bibinfo}[2]{#2}
\providecommand{\BIBentrySTDinterwordspacing}{\spaceskip=0pt\relax}
\providecommand{\BIBentryALTinterwordstretchfactor}{4}
\providecommand{\BIBentryALTinterwordspacing}{\spaceskip=\fontdimen2\font plus
\BIBentryALTinterwordstretchfactor\fontdimen3\font minus
  \fontdimen4\font\relax}
\providecommand{\BIBforeignlanguage}[2]{{%
\expandafter\ifx\csname l@#1\endcsname\relax
\typeout{** WARNING: IEEEtran.bst: No hyphenation pattern has been}%
\typeout{** loaded for the language `#1'. Using the pattern for}%
\typeout{** the default language instead.}%
\else
\language=\csname l@#1\endcsname
\fi
#2}}
\providecommand{\BIBdecl}{\relax}
\BIBdecl

\bibitem{shannon}
C.~E. Shannon, ``A mathematical theory of communication,'' \emph{Bell Sys.
  Tech. J.}, vol.~27, pp. 379--423, 623--656, 1948.

\bibitem{huffman}
D.~A. Huffman \emph{et~al.}, ``A method for the construction of minimum
  redundancy codes,'' \emph{proc. IRE}, vol.~40, no.~9, pp. 1098--1101, 1952.

\bibitem{lzcomp}
J.~Ziv and A.~Lempel, ``A universal algorithm for sequential data
  compression,'' \emph{IEEE Transactions on information theory}, vol.~23,
  no.~3, pp. 337--343, 1977.

\bibitem{pavlichin2013human}
D.~S. Pavlichin, T.~Weissman, and G.~Yona, ``The human genome contracts
  again,'' \emph{Bioinformatics}, vol.~29, no.~17, pp. 2199--2202, 2013.

\bibitem{deorowicz2013genome}
S.~Deorowicz, A.~Danek, and S.~Grabowski, ``Genome compression: a novel
  approach for large collections,'' \emph{Bioinformatics}, vol.~29, no.~20, pp.
  2572--2578, 2013.

\bibitem{relational}
E.~F. Codd, ``A relational model of data for large shared data banks,''
  \emph{Communications of the ACM}, vol.~13, no.~6, pp. 377--387, 1970.

\bibitem{elias}
P.~Elias and R.~A. Flower, ``The complexity of some simple retrieval
  problems,'' \emph{Journal of the ACM}, vol.~22, no.~3, pp. 367--379,
  1975.

\bibitem{1993bit}
P.~Miltersen, \emph{The bit probe complexity measure revisited}.\hskip 1em
  plus 0.5em minus 0.4em\relax \!\!\! \!\!\! Springer, 1993.

\bibitem{cellprobe}
------, ``Cell probe complexity-a survey,'' in \emph{19th Conference on the
  Found. of Software Tech. and Theoretical Computer Science}, 1999.

\bibitem{yao}
A.~C.-C. Yao, ``Should tables be sorted?'' \emph{Journal of the ACM (JACM)},
  vol.~28, no.~3, pp. 615--628, 1981.

\bibitem{bitvector}
H.~Buhrman, P.~B. Miltersen, J.~Radhakrishnan, and S.~Venkatesh, ``Are
  bitvectors optimal?'' \emph{SIAM Journal on Computing}, vol.~31, no.~6, pp.
  1723--1744, 2002.

\bibitem{nicholson2013survey}
P.~K. Nicholson, V.~Raman, and S.~S. Rao, ``A survey of data structures in the
  bitprobe model,'' in \emph{Space-Efficient Data Structures, Streams, and
  Algorithms}.\hskip 1em plus 0.5em minus 0.4em\relax Springer, 2013, pp.
  303--318.

\bibitem{alon2009power}
N.~Alon and U.~Feige, ``On the power of two, three and four probes,'' in
  \emph{Proceedings of the twentieth Annual ACM-SIAM Symposium on Discrete
  Algorithms}.\hskip 1em plus 0.5em minus 0.4em\relax SIAM, 2009, pp. 346--354.

\bibitem{radhakrishnan2001explicit}
J.~Radhakrishnan, V.~Raman, and S.~Srinivasa~Rao, ``Explicit deterministic
  constructions for membership in the bitprobe model,'' \emph{Lecture Notes in
  Computer Science}, vol. 2161, pp. 290--299, 2001.

\bibitem{viola2012bit}
E.~Viola, ``Bit-probe lower bounds for succinct data structures,'' \emph{SIAM
  Journal on Computing}, vol.~41, no.~6, pp. 1593--1604, 2012.

\bibitem{lewenstein2014improved}
M.~Lewenstein, J.~I. Munro, P.~K. Nicholson, and V.~Raman, ``Improved explicit
  data structures in the bitprobe model,'' in \emph{Algorithms-ESA 2014}.\hskip
  1em plus 0.5em minus 0.4em\relax Springer, 2014, pp. 630--641.

\bibitem{chandar}
V.~Chandar, D.~Shah, and G.~W. Wornell, ``A locally encodable and decodable
  compressed data structure,'' in \emph{Communication, Control, and Computing,
  2009. Allerton 2009. 47th Annual Allerton Conference on}.\hskip 1em plus
  0.5em minus 0.4em\relax IEEE, 2009, pp. 613--619.

\bibitem{makhdoumi2013locally}
A.~Makhdoumi, S.-L. Huang, Y.~Polyanskiy, and M.~Medard, ``On locally decodable
  source coding,'' \emph{arXiv preprint arXiv:1308.5239}, 2013.

\bibitem{zhou2014simple}
H.~Zhou, D.~Wang, and G.~Wornell, ``A simple class of efficient compression
  schemes supporting local access and editing,'' in \emph{Information Theory
  (ISIT), 2014 IEEE International Symposium on}.\hskip 1em plus 0.5em minus
  0.4em\relax IEEE, 2014, pp. 2489--2493.

\bibitem{katz2000efficiency}
J.~Katz and L.~Trevisan, ``On the efficiency of local decoding procedures for
  error-correcting codes,'' in \emph{Proceedings of the thirty-second annual
  ACM symposium on Theory of computing}.\hskip 1em plus 0.5em minus 0.4em\relax
  ACM, 2000, pp. 80--86.


\bibitem{miltersen1995data}
P.~B. Miltersen, N.~Nisan, S.~Safra, and A.~Wigderson, ``On data structures and
  asymmetric communication complexity,'' in \emph{Proceedings of the
  twenty-seventh annual ACM symposium on Theory of computing}.\hskip 1em plus
  0.5em minus 0.4em\relax ACM, 1995, pp. 103--111.

\bibitem{LYM}
K.~Yamamoto \emph{et~al.}, ``Logarithmic order of free distributive lattice,''
  \emph{Jnl. of the Mathematical Soc. of Japan}, vol.~6, no. 3-4, pp.
  343--353, 1954.

\end{thebibliography}

\appendices
\section{Supplementary Material}\label{sec:proofs}
In this section, we provide further details for the proof sketches given in Section \ref{sec:MainResults}.  For convenience, we recall the following standard inequalities which will be used repeatedly throughout the proofs without explicit mention:
\begin{align}
\left( \frac{n}{k}\right)^k \leq {n \choose k} \leq \frac{n^k}{k!} \leq \left( \frac{n e}{k}\right)^k.
\end{align}

\subsection{Details for proof of theorem \ref{lb:ad}} \label{app:thm1}
The LYM inequality used in proving the lower bound \eqref{nonAsympConverse} is given below for convenience.

\begin{lemma}[LYM inequality \cite{LYM}]
Let ${U}$ be a $u$-element set, let $\mathcal{A}$ be a family of subsets of $U$ such that no set in $\mathcal{A}$ is a subset of another set in $\mathcal{A}$, and let $a_m$ denote the number of sets of size $m$ in $\mathcal{A}$. Then
\begin{align}
\sum_{m=0}^u\frac{a_m}{{u \choose m}} \le 1. \label{LYMin}
\end{align}
\end{lemma}

The second key inequality was Lemma \ref{lem:techEstimate}, which is restated below for convenience.

\noindent {\bf Lemma \ref{lem:techEstimate}.}\emph{
For all $M,v\geq 1$, 
\begin{align}
\sum_{k=1}^M \max_{i \leq v} {{2k} \choose i } \leq 
2^{v} \frac{\left(M+2 + \frac{v+1}{2e}\right)^{v+1}}{(v+1)!}. \label{eqn:KeyConverseLemma2}%
\end{align}
}
\begin{IEEEproof}[Proof of Lemma  \ref{lem:techEstimate}]
We begin the proof  by splitting the bound into two cases:
\begin{proposition}\label{prop:TwoBounds}
For $v\geq 1$
\begin{align}
\sum_{k=1}^M \max_{i \leq v} {{2k} \choose i } \leq 
\begin{cases}
\frac{3}{2} {2M \choose M} & \mbox{for $1 \leq M \leq v$}\\
 2^{v} \frac{(M+2)^{v+1}}{(v+1)!} & \mbox{for $M \geq v+1$}.
\end{cases}
\end{align}
\end{proposition}
\begin{proof}
Note that if $M\leq v$, then 
\begin{align}
\sum_{k=1}^M \max_{i \leq v} {{2k} \choose i } =\sum_{k=1}^M  {{2k} \choose k }. 
\end{align}
We will prove by induction on $M$ that $\sum_{k=1}^M  {{2k} \choose k }\leq \frac{3}{2}{{2M} \choose M}$.    The base case $M=1$ is trivial, so by the inductive hypothesis, we have
\begin{align}
\sum_{k=1}^{M+1}  {{2k} \choose k } \leq {{2(M+1)}\choose M+1} + \frac{3}{2}  {{2M} \choose M }.
\end{align}
However, we can  write
\begin{align}
{{2(M+1)}\choose M+1} = \frac{(2M+2)(2M+1)}{(M+1)^2} {{2M} \choose M }\geq 3  {{2M} \choose M },\notag
\end{align}
completing the proof of the first claim.

Next, if $M\geq v+1$, then 
\begin{align}
\sum_{k=1}^{M} \max_{i\leq v}{{2k} \choose i}  
&=\sum_{k=1}^{v}  {2k \choose k} +  \sum_{k=v+1}^{M} {2k \choose v}\\
&\leq \sum_{k=v}^{M+1} {2k \choose v}\\
&\leq \sum_{k=v}^{M+1} \frac{(2k)^v}{v!}\\
&\leq  \frac{2^{v}}{ v!} \int_{0}^{M+2} z^{v} dz \\
&\leq 2^{v} \frac{  ( M + 2 )^{ v  + 1}   }{(v+1)!},%
\end{align}
establishing the second claim.
\end{proof}

In light of Proposition \ref{prop:TwoBounds}, it is sufficient to check that 
\begin{align}
{(3(v+1)!)}^{1/(v+1)} {2M \choose M}^{1/(v+1)} \leq 2 \left(M+2 + \frac{v+1}{2e}\right)  \label{UBtoShow}
\end{align}
for $1\leq M \leq v$ in order to prove \eqref{eqn:KeyConverseLemma2}.  Since 
\begin{align}
{(3(v+1)!)}^{1/(v+1)} {2M \choose M}^{1/(v+1)} \!\!\! \leq 3^{1/(v+1)} \frac{v+1}{e} \left(2e\right)^{M/(v+1)}\label{cvxRHS}
\end{align}
and the RHS in \eqref{cvxRHS} is a convex function in $M$, we can show \eqref{UBtoShow} holds by verifying 
\begin{align}
3^{1/(v+1)} \frac{v+1}{e} \left(2e\right)^{M/(v+1)} \leq 2 \left(M+2 + \frac{v+1}{2e}\right) 
\end{align}
for $M=0$ and $M=v+1$.  It is straightforward to check that this is the case.
\end{IEEEproof}

\subsection{Details for proof of Theorem \ref{ub:nonad}} \label{app:thm2}

\begin{figure}%
\center
\scalebox{0.65}{
\begin{tikzpicture}
  \tikzstyle{bigarrow}=[decoration={markings,mark=at position 0.999 with
    {\arrow[scale=2]{stealth}}}, postaction={decorate}, shorten >=0.4pt]
  
  \draw[fill=gray] (0,-1) rectangle (1,-2);
  \draw[fill=gray] (0,-2) rectangle (1,-3);
  \draw[rounded corners, thick] (0,0) rectangle (1,-12);
  \foreach \y in {-1,...,-11} \draw (0, \y) -- (1, \y);

  \begin{scope}[xshift = 1.5cm]
  \draw[fill=gray] (0,-1) rectangle (1,-2);
  \draw[fill=gray] (0,-4) rectangle (1,-5);
  \draw[rounded corners, thick] (0,0) rectangle (1,-12);
  \foreach \y in {-1,...,-11} \draw (0, \y) -- (1, \y);
  \end{scope}
  
  \begin{scope}[xshift = 3cm]
  \draw[fill=gray] (0,-1) rectangle (1,-2);
  \draw[fill=gray] (0,-5) rectangle (1,-6);
  \draw[rounded corners, thick] (0,0) rectangle (1,-12);
  \foreach \y in {-1,...,-11} \draw (0, \y) -- (1, \y);
  \end{scope}
  
  \begin{scope}[xshift = 4.5cm]
  \draw[fill=gray] (0,-2) rectangle (1,-3);
  \draw[fill=gray] (0,-4) rectangle (1,-5);
  \draw[rounded corners, thick] (0,0) rectangle (1,-12);
  \foreach \y in {-1,...,-11} \draw (0, \y) -- (1, \y);
  \end{scope}
  
  \begin{scope}[xshift = 6cm]
  \draw[fill=gray] (0,-2) rectangle (1,-3);
  \draw[fill=gray] (0,-5) rectangle (1,-6);
  \draw[rounded corners, thick] (0,0) rectangle (1,-12);
  \foreach \y in {-1,...,-11} \draw (0, \y) -- (1, \y);
  \end{scope}
  
  \begin{scope}[xshift = 7.5cm]
  \draw[fill=gray] (0,-4) rectangle (1,-5);
  \draw[fill=gray] (0,-5) rectangle (1,-6);
  \draw[rounded corners, thick] (0,0) rectangle (1,-12);
  \foreach \y in {-1,...,-11} \draw (0, \y) -- (1, \y);
  \end{scope}
  
  \begin{scope}[xshift = 9.5cm]
  \draw[dashed] (0,-1) rectangle (1,-2);
  \begin{scope}[yshift = -1cm]
	\draw[dashed] (0,-1) rectangle (1,-2);
  \end{scope}
  \begin{scope}[yshift = -3cm]
	\draw[dashed] (0,-1) rectangle (1,-2);
  \end{scope}
  \begin{scope}[yshift = -4cm]
	\draw[dashed] (0,-1) rectangle (1,-2);
  \end{scope}
  \end{scope}

  \begin{scope}[xshift = 2cm, yshift = -14cm]
  \begin{scope}[xshift = 0cm]
  \end{scope}
  \draw[fill=gray] (1,0) rectangle (2,-1);
  \draw[fill=gray] (2,0) rectangle (3,-1);
  \draw[fill=gray] (3,0) rectangle (4,-1);
  \draw[fill=gray] (5,0) rectangle (6,-1);
  \draw[fill=gray] (6,0) rectangle (7,-1);
  \draw[fill=gray] (7,0) rectangle (8,-1);
  \draw[rounded corners, thick] (0,0) rectangle (10,-1);
  \foreach \x in {1,...,9} \draw (\x, 0) -- (\x, -1);
  \end{scope}  
  
  \begin{scope}[xshift = 2cm, yshift = -15.5cm]
  \begin{scope}[xshift = 0cm]
  \end{scope}
  \draw[fill=gray] (1,0) rectangle (2,-1);
  \draw[fill=gray] (3,0) rectangle (4,-1);
  \draw[fill=gray] (4,0) rectangle (5,-1);
  \draw[fill=gray] (5,0) rectangle (6,-1);
  \draw[fill=gray] (6,0) rectangle (7,-1);
  \draw[fill=gray] (7,0) rectangle (8,-1);
  \draw[rounded corners, thick] (0,0) rectangle (10,-1);
  \foreach \x in {1,...,9} \draw (\x, 0) -- (\x, -1);
  \end{scope}
  
  \begin{scope}[xshift = 2cm, yshift = -17cm]
  \begin{scope}[xshift = 0cm]
  \end{scope}
  \draw[fill=gray] (2,0) rectangle (3,-1);
  \draw[fill=gray] (3,0) rectangle (4,-1);
  \draw[fill=gray] (4,0) rectangle (5,-1);
  \draw[fill=gray] (5,0) rectangle (6,-1);
  \draw[fill=gray] (6,0) rectangle (7,-1);
  \draw[fill=gray] (7,0) rectangle (8,-1);
  \draw[rounded corners, thick] (0,0) rectangle (10,-1);
  \foreach \x in {1,...,9} \draw (\x, 0) -- (\x, -1);
  \end{scope}
    
  \node[scale=1.5] at (0.5,0.5) {$x^n_{(1)}$};
  \node[scale=1.5] at (2,0.5) {$x^n_{(2)}$};
  \node[scale=1.5] at (3.5,0.5) {$x^n_{(3)}$};
  \node[scale=1.5] at (5,0.5) {$x^n_{(4)}$};
  \node[scale=1.5] at (6.5,0.5) {$x^n_{(5)}$};
  \node[scale=1.5] at (8,0.5) {$x^n_{(6)}$};
  \node[scale=1.5] at (10,0.5) {$S_k$};
  \node[scale=1.5] at (12,0.5) {$T_{j,k}$};
  
  \node[scale=1.2] at (1,-14.5) {$\enc(x^n_{(3)})$};
  \node[scale=1.2] at (1,-16) {$\enc(x^n_{(5)})$};
   \node[scale=1.2] at (1,-17.5) {$\enc(x^n_{(6)})$};
    
  \node[scale=1.2] at (12,-1.5) {$\{2,3,4\}$};
    \node[scale=1.2] at (12,-2.5) {$\{2,4,5\}$};
      \node[scale=1.2] at (12,-4.5) {$\{3,4,5\}$};
        \node[scale=1.2] at (12,-5.5) {$\{6,7,8\}$};

\end{tikzpicture}
}
\caption{Encoding of  $2$-sparse source sequences of length $n=12$ using a codeword of length $k=10$, with at most $d=3$ bit probes. As a result, $|S_k| = \frac{r}{r+1} {\binom{k}{d}}/{\binom{rd}{d}} = 4$. The source sequences $x^n_{(1)}, \ldots, x^n_{(6)}$ represent those that satisfy condition {(C1)} of the encoding criterion for $S_k = \{2,3,5,6\}$. Condition {(C2)} is only satisfied by $x^n_{(3)}, x^n_{(5)}, x^n_{(6)}$, which are encoded as shown. If the bit to be decoded $j \in [n] \setminus S_k = \{1,4,7,8,9,10,11,12\}$, then the decoder outputs $0$ without probing the bits of the codeword. If $j\in S_k$, then the decoder probes positions $T_{j,k}$ of the codeword and returns the AND of the bits (shaded blocks correspond to 1's, unshaded blocks signify 0's). } \label{fig:r2}
\end{figure}

The random encoding scheme described in Section \ref{sec:explength} is illustrated in Figure \ref{fig:r2}.  The only thing remaining is to analyze the performance of this scheme:

{\bf Performance Analysis:}   
 To show a bound on the expected codeword length, fix an arbitrary sequence $x^n$ and define the events
\begin{align*}
\mathcal{E}_{k,1} &= \{ \mathsf{supp}(x^n)\subseteq S_k \}\\
\mathcal{E}_{k,2} &= \{ T_{j,k}\not\subseteq \cup_{ i \in \mathsf{supp}(x^n) } T_{i,k} \mbox{~~for all~~} j \in S_k \setminus \mathsf{supp}(x^n) \}.
\end{align*}
By independence of the sets used in the codebook construction, we have 
\begin{align}
&\EE_{\mathcal{C}}\left[ \ell(\mathsf{c}(x^n))\right] \label{codebookEnsemble}\\
&= \sum_{k\geq rd+1} k \Pr\{ \mathcal{E}_{k,1} \cap \mathcal{E}_{k,2}\} \prod_{j=rd+1}^{k-1}\left(1  - \Pr\{ \mathcal{E}_{j,1} \cap \mathcal{E}_{j,2}\} \right),\notag
\end{align} 
where $\EE_{\mathcal{C}}\left[ \cdot \right]$ denotes expectation over the ensemble of random codebooks.  
Importantly, we note that \eqref{codebookEnsemble} is a decreasing function of $\Pr\{ \mathcal{E}_{k,1} \cap \mathcal{E}_{k,2}\}$ for each $k$.  Therefore, in order to upper bound \eqref{codebookEnsemble}, we will  lower bound  $\Pr\{ \mathcal{E}_{k,1} \cap \mathcal{E}_{k,2}\}$.  To that end, observe that 
\begin{align}
\Pr\{ \mathcal{E}_{k,1} \cap \mathcal{E}_{k,2}\} &= \Pr\{ \mathcal{E}_{k,1} \} \Pr\{ \mathcal{E}_{k,2} | \mathcal{E}_{k,1}\} \\
&= \frac{ {{|S_k|} \choose r }  }{  {n \choose r} } \Pr\{ \mathcal{E}_{k,2} | \mathcal{E}_{k,1}\},
\end{align}
where the conditional probability $\Pr\{ \mathcal{E}_{k,2} | \mathcal{E}_{k,1}\}$ can be bounded from below by a simple union bound:
\begin{align}
\Pr\{ \mathcal{E}_{k,2} | \mathcal{E}_{k,1}\} &\geq 1 -\!\!\!\!\! \sum_{ j \in S_k \setminus \mathsf{supp}(x^n) }  \Pr\{  T_{j,k} \subseteq \cup_{ i \in \mathsf{supp}(x^n) } T_{i,k}  \}\notag\\
&\geq 1 -  |S_k|  \frac{{rd \choose d} }{  {k \choose d}  }.
\end{align}
Therefore, we have
\begin{align}
\Pr\{ \mathcal{E}_{k,1} \cap \mathcal{E}_{k,2}\} &\geq  \frac{1}{{n \choose r}} {{|S_k|} \choose r } \left(   1 -  |S_k|  \frac{{rd \choose d} }{  {k \choose d}  } \right) \\
&=   \frac{1}{{n \choose r}} \frac{1}{r+1}\binom{\frac{r}{r+1}\frac{\binom{k}{d}}{\binom{rd}{d}}}{r}.\label{qtyToBound}
\end{align}
The challenge of the proof is to now carefully bound \eqref{qtyToBound} and \eqref{codebookEnsemble}.
Toward this goal, we further bound $\Pr\{ \mathcal{E}_{k,1} \cap \mathcal{E}_{k,2}\}$ as follows:
\begin{align}
{n \choose r}  \Pr\{ \mathcal{E}_{k,1} \cap \mathcal{E}_{k,2}\} 
&\geq \frac{1}{r+1}\binom{\frac{r}{r+1}\frac{\binom{k}{d}}{\binom{rd}{d}}}{r} \\
&\geq \frac{1}{(r+1)^{r+1}}\left( \frac{\binom{k}{d}}{\binom{rd}{d}}  \right)^r\\
&\geq \frac{1}{(r+1)^{r+1}}\left( \frac{ \left(\frac{k}{d}\right)^d }{\left(\frac{rde}{d}\right)^d}  \right)^r\\
 &=\frac{k^{rd}}{(r+1)^{r+1}(e rd) ^{rd}}. %
\end{align}
Hence, 
\begin{align}
\Pr\{ \mathcal{E}_{k,1} \cap \mathcal{E}_{k,2}\} \geq C_{r,d}  {k^{rd}} ,\label{defnCrd}
\end{align}
where we have defined $C_{r,d} = \left({n \choose r }  (r+1)^{r+1}(e rd) ^{rd} \right)^{-1}$ for convenience.  

Now, using the inequality $(1-x) \leq e^{-x}$, we can upper bound \eqref{codebookEnsemble} with \eqref{defnCrd} as
\begin{align}
&\EE_{\mathcal{C}}\left[ \ell(\mathsf{c}(x^n))\right]  \notag\\
&\leq  \sum_{k\geq rd+1} C_{r,d}  {k^{rd+1}} \exp\left(-\sum_{j=rd+1}^{k-1} C_{r,d}  {j^{rd}} \right) \\
&\leq  \sum_{k\geq rd+1} C_{r,d}  {k^{rd+1}} \exp\left(- C_{r,d} \int_{rd+1}^{k-1} z^{rd} dz   \right) \\
&= \sum_{k\geq rd+1} \Bigg[ C_{r,d}  {k^{rd+1}}  \times \notag \\
&~~~\exp\left(- C_{r,d}\frac{\left((k-1)^{rd+1}-(rd+1)^{rd+1} \right)}{rd+1}   \right) \Bigg]\\
&=\exp\left( C_{r,d} (rd+1)^{rd}   \right)  \times \notag\\
&~~~\sum_{k\geq rd+1} C_{r,d}   {k^{rd+1}}  \exp\left(- C_{r,d}\frac{(k-1)^{rd+1} }{rd+1}   \right) \notag\\
&\leq  \exp\left( C_{r,d} (rd+1)^{rd}  +\frac{rd+1}{rd} \right) \times \notag\\
&~~~ \sum_{k\geq rd+1} C_{r,d}  {(k-1)^{rd+1}} \exp\left(- C_{r,d}\frac{(k-1)^{rd+1} }{rd+1}   \right)\\
&=   \exp\left( C_{r,d} (rd+1)^{rd}  +\frac{rd+1}{rd} \right)  \times \notag\\
&~~~(rd+1) \sum_{k\geq rd} C_{r,d}  \frac{k^{rd+1}}{rd+1} \exp\left(- C_{r,d}\frac{k^{rd+1} }{rd+1}   \right).\label{sumToBound} %
\end{align}
Since the function $u e^{-u}$ is monotone increasing on $(0,1)$ and monotone decreasing on $(1,\infty)$ with a maximum of $1/e$, we can bound the sum in \eqref{sumToBound} as 
\begin{align}
&\sum_{k=rd}^{\infty} C_{r,d}  \frac{k^{rd+1}}{rd+1} \exp\left(- C_{r,d}\frac{k^{rd+1} }{rd+1}   \right) \notag \\
&\leq 2/e +  \int_{0}^{\infty} C_{r,d}   \frac{ z^{rd+1}}{rd+1  } \exp\left(- C_{r,d} \frac{ z^{rd+1}}{rd+1 }  \right) dz. \label{intToBound}
\end{align}

The integral in \eqref{intToBound} can be bounded as follows:
\begin{lemma}\label{lem:BoundIntegral}
\begin{align*}
&(rd+1) \int_{0}^{\infty} C_{r,d}   \frac{ z^{rd+1}}{rd+1  } \exp\left(- C_{r,d} \frac{ z^{rd+1}}{rd+1  }  \right) dz  \\
&\quad\leq  \left(\frac{rd+1}{C_{r,d}}\right)^{1/(rd+1)} . 
\end{align*}
\end{lemma}
\begin{proof}
Abbreviating $C:=C_{r,d}$, consider the change of variables $z = \Big(\frac{u (rd+1)}{C}\Big)^{1/(rd+1)}$.  Then,
\begin{align}
&(rd+1)\int_{0}^{\infty} C   \frac{ z^{rd+1}}{rd+1  } \exp\left(- C \frac{ z^{rd+1}}{rd+1  }  \right) dz\\
&= (rd+1)\int_{0}^{\infty} u  e^{-u} dz\\
&= (rd+1)\int_{0}^{\infty} \!u  e^{-u} \!\left( \frac{u^{-rd/(rd + 1) }}{rd+1} \left(\frac{rd+1}{C}\right)^{1/(rd+1)} du \!\right)\notag \\
&=  \left(\frac{rd+1}{C}\right)^{1/(rd+1)}  \int_{0}^{\infty} u^{1/(rd+1)}  e^{-u}  du\\
&=   \left(\frac{rd+1}{C}\right)^{1/(rd+1)}  \Gamma\left( \frac{rd+2}{rd+1}\right)\\
&= \left(\frac{rd+1}{C}\right)^{1/(rd+1)}    \Gamma\left( \frac{rd+2}{rd+1}\right)\notag \\
&\leq \left(\frac{rd+1}{C}\right)^{1/(rd+1)}  , \label{gammaBound}
\end{align}
where \eqref{gammaBound} follows because $\Gamma(x)\leq 1$ for $x\leq 2$.
\end{proof}

To finish the proof, we use Lemma \ref{lem:BoundIntegral} and the integral upper bound on  \eqref{sumToBound} to conclude 
\begin{align}
&\EE_{\mathcal{C}}\left[ \ell(\mathsf{c}(x^n))\right]  \notag\\
&\leq 2 (rd+1) \exp\left( C_{r,d} (rd+1)^{rd}  +\frac{1}{rd} \right) +\notag \\
& ~~~\exp\left( C_{r,d} (rd+1)^{rd}  +\frac{rd+1}{rd} \right) \left(\frac{rd+1}{C_{r,d}}\right)^{1/(rd+1)} \\
&\leq   30 (rd+1)  \left((r+1)^{(r+1)} {n \choose r }\right)^{1/(rd+1)}   \label{lastline}, 
\end{align}
where \eqref{lastline} is a (loose) upper bound assuming $rd\geq 1$,  which follows from elementary algebra and the definition of $C_{r,d}$.

Since all sequences in ${{[n]} \choose r}$ are equally probable, linearity of expectation ensures the existence of a code $\enc$ which satisfies \eqref{nonAsympAch} as desired.

\begin{remark}
We have made no significant attempt to optimize the multiplicative constant in \eqref{nonAsympAch}.  In general, the given proof reveals that this factor of 30 can be replaced by a function of $rd$ that is bounded by 30 when $rd=1$ and is upper bounded by $2+e$ as $rd$ grows large.  We conjecture that the achievability scheme proposed in the proof of Theorem \ref{ub:nonad} can yield a multiplicative factor as small as $\Gamma\left(\frac{rd+2}{rd+1}\right)$, which is strictly less than 1,  at the expense of more careful intermediate bounds. %
\end{remark}

\subsection{Proof of Theorem \ref{thm:compOptimalNecessary}}

Suppose $\{\enc_n\}$ is a competitively optimal sequence of $(r_n,d_n,n)$-locally decodable codes, and define $\epsilon_n$ according to
\begin{align}
r_n d_n +1 = \epsilon_n \log {n \choose r_n}.
\end{align}
  Then Theorem \ref{lb:ad} and the definition of competitive optimality imply that there is some constant $K$ for which 
\begin{align}
K &\geq \frac{\EE[\ell(\enc_n(X^n))]}{\log {n \choose r_n}}\\
&\geq {\epsilon_n }\left( {n \choose r_n}^{1/(\epsilon_n \log{n \choose r_n})}-1 \right)\\
&= {\epsilon_n }\left(2^{1/\epsilon_n}-1 \right)
\end{align}  
for all $n$ sufficiently large.
Since $\epsilon({2^{1/\epsilon}-1})\nearrow \infty$ as $\epsilon\searrow 0$, this implies that there is a constant $K'>0$ such that $\epsilon_n\geq K'$ for all $n$ sufficiently large, proving the  claim.

\subsection{Proof of Corollary \ref{cor:CompOptimalSparse}}

To prove the ``only if" direction, note that Theorem \ref{thm:compOptimalNecessary} asserts that there must be a constant $K>0$ such that 
\begin{align}
r_n d_n \geq K \log {n \choose r_n } \geq K r_n \log {n \over r_n }.
\end{align}
Thus, by the assumption that $r_n = O(n^{1-\epsilon})$ for some $\epsilon>0$, 
\begin{align}
d_n \geq K \log {n \over r_n } = \Omega (\log n).
\end{align}

To prove the ``if" direction, suppose there are positive constants $K,\epsilon$ such that for $n$ sufficiently large
\begin{align}
d_n &\geq K \log {n }\\
r_n &\leq  n^{1-\epsilon}.
\end{align}
Since reducing $d_n$ can only adversely affect performance, we can assume without loss of generality that 
\begin{align}
d_n +1 \leq  2 K \log {n } \leq \frac{2 K}{\epsilon} \log{n \over r_n}
\end{align}
for $n$ sufficiently large.  As a consequence, we have
\begin{align}
r_n d_n n^{r_n/(r_n d_n +1)} &\leq r_n d_n n ^{1/( d_n +1)} \notag\\
&\leq  2^{1/2K} \frac{2 K}{\epsilon} r_n \log{n \over r_n} \\
&\leq 2^{1/2K} \frac{2 K}{\epsilon} \log {n \choose r_n}.
\end{align}
An application of  Theorem \ref{ub:nonad} completes the proof.

\end{document}